\documentclass[sigconf]{acmart}

\usepackage{booktabs} 
\usepackage{subcaption} 
\usepackage{algorithm}
\usepackage{algorithmic}

\usepackage[capitalize,noabbrev]{cleveref}

\acmConference[KDD '26]{Proceedings of the 32nd ACM SIGKDD Conference on Knowledge Discovery and Data Mining}{August 2026}{Vienna, Austria}
\acmYear{2026}
\copyrightyear{2026}
\acmBooktitle{Proceedings of the 32nd ACM SIGKDD Conference on Knowledge Discovery and Data Mining (KDD '26), August 2026, Vienna, Austria}

\begin{document}

\title{EWSJF: An Adaptive Scheduler with Hybrid Partitioning for Mixed-Workload LLM Inference}

\author{Bronislav Sidik}
\email{slava.sidik@huawei.com}
\affiliation{%
  \institution{Toga Networks (Huawei)}
  \country{Israel}
}

\author{Chaya Levi}
\email{chaya.levi@huawei.com}
\affiliation{%
  \institution{Toga Networks (Huawei)}
  \country{Israel}
}

%
\author{Joseph Kampeas}
\email{joseph.kampeas@huawei.com}
\affiliation{%
  \institution{Toga Networks (Huawei)}
  \country{Israel}
}
\begin{abstract}
Serving Large Language Models (LLMs) under mixed workloads—short, latency-sensitive interactive queries alongside long, throughput-oriented batch requests—poses a fundamental scheduling challenge. Standard First-Come, First-Served (FCFS) policies suffer from severe head-of-line blocking, leading to high tail latency and underutilized hardware. We introduce \textbf{EWSJF} (Effective Workload-based Shortest Job First), an adaptive request-level scheduler that learns workload structure in real time to jointly improve fairness and throughput. EWSJF operates upstream of execution-level schedulers and integrates four components: (1) \textit{Refine-and-Prune}, an unsupervised partitioning algorithm that discovers performance-homogeneous request groups; (2) \textit{Dynamic Queue Routing} for assigning requests to these groups; (3) \textit{Density-Weighted Scoring}, a context-aware prioritization function balancing urgency and fairness; and (4) \textit{Bayesian Meta-Optimization}, which continuously tunes scoring and partitioning parameters based on live performance feedback. Implemented in vLLM, EWSJF improves end-to-end throughput by over 30\% and reduces average Time-To-First-Token for short requests by up to 4$\times$ compared to FCFS. These results demonstrate that adaptive, learning-based request scheduling is a critical missing layer for efficient and responsive LLM serving.Implementation available at \url{https://anonymous.4open.science/r/vllm_0110-32D8}.
\end{abstract}

\ccsdesc[500]{Computing methodologies~Artificial intelligence}
\ccsdesc[500]{Computer systems organization~Distributed architectures}

\keywords{LLM Inference, Scheduling, Queue Management, Adaptive Systems}

\maketitle

\section{Introduction}
\label{sec:introduction}

Large Language Models (LLMs) have shifted the primary computational bottleneck from training to inference, making efficient serving systems essential for real-world deployment. Frameworks such as vLLM~\cite{kwon2023efficient} significantly improve throughput through techniques like PagedAttention, yet they are optimized for relatively uniform workloads. In practice, production LLM services increasingly operate under \textit{mixed workloads}: short, latency-sensitive interactive queries (e.g., chatbots) arriving alongside long, throughput-oriented batch jobs (e.g., summarization). This heterogeneity exposes fundamental limitations in standard request-level scheduling.

The default First-Come, First-Served (FCFS) scheduler suffers from severe head-of-line (HoL) blocking, where a single long request delays all others. Under concurrent load, this leads to extreme spikes in Time-To-First-Token (TTFT) for short queries—sometimes exceeding minutes~\cite{vllm_issue_3096}—forcing operators to choose between degraded user experience or underutilized hardware. 

Existing solutions fall short. Static priority queues provide only coarse separation and cannot adapt to shifting workload distributions. Academic schedulers such as Orca~\cite{sheng2024orca} and G-Fair~\cite{li2024fifo} target theoretical optimality or fairness in structured, multi-tenant environments, but do not address the dynamic, unstructured nature of real-world LLM traffic. Moreover, execution-level schedulers like Orca~\cite{sheng2024orca} and Sarathi~\cite{agrawal2024sarathi} operate \emph{within} the decoding loop, optimizing token-level execution after batches have already been formed. Modern engines such as vLLM already incorporate many of these iteration-level optimizations (e.g., continuous batching), leaving the upstream problem of \emph{request admission} largely unaddressed.

\textbf{EWSJF} (Effective Workload-based Shortest Job First) fills this gap. Rather than modifying the execution engine, EWSJF operates \emph{above} it as an intelligent admission controller. Its goal is to organize heterogeneous requests into performance-homogeneous batches, enabling downstream execution-level schedulers to operate more efficiently while avoiding the fairness and starvation issues of naive SJF-style heuristics.

EWSJF is built around three core ideas: (1) \textbf{Refine-and-Prune}, an unsupervised algorithm that discovers performance-homogeneous request groups; (2) \textbf{Density-Weighted Scoring}, a context-aware prioritization function balancing urgency and fairness; and (3) a \textbf{Bayesian Meta-Optimizer} that continuously tunes scheduling parameters based on live performance feedback.

Our contributions are as follows:
\begin{itemize}
    \item We introduce EWSJF, a novel adaptive request-level scheduler for LLM inference that autonomously learns and adapts to workload structure.
    \item We propose \textit{Refine-and-Prune}, a hybrid partitioning algorithm that constructs performance-homogeneous request queues without requiring a fixed number of clusters.
    \item We develop a context-aware, learnable scoring framework that balances latency, fairness, and computational cost through density-weighted prioritization.
    \item We implement EWSJF in vLLM and demonstrate substantial improvements in both throughput and user-perceived latency under realistic mixed workloads.
\end{itemize}

To facilitate reproducibility and community adoption, we provide our full implementation as a pluggable module for vLLM, along with the dataset generation scripts, in our anonymous repository~\cite{ewsjf_impl}
\section{Related Work}
\label{sec:related_work}

Scheduling for LLM inference has rapidly emerged as a critical research area, driven by the need to balance throughput, latency, and fairness under heterogeneous workloads. EWSJF introduces a learning-based approach that is distinct from, yet complementary to, existing strategies. Prior work can be grouped into four categories: throughput-oriented schedulers, fairness-driven schedulers, agent-level schedulers, and execution co-design frameworks. Crucially, EWSJF complements iteration-level schedulers like Orca~\cite{sheng2024orca} by optimizing request admission rather than execution. It serves as an upstream policy that creates homogeneous batches. We therefore evaluate against request-level baselines (FCFS, SJF) to isolate partitioning gains. 

\subsection{LLM Scheduling Policies}

Table~\ref{tab:scheduler_comparison} summarizes key differences between EWSJF and representative state-of-the-art systems.

\paragraph{Throughput-Optimal Schedulers.}
Orca~\cite{sheng2024orca} and Sarathi~\cite{agrawal2024sarathi} pursue near-optimal throughput via \emph{iteration-level scheduling}, enabling preemption and token-level control. While theoretically powerful, these approaches require deep integration with the execution engine. EWSJF instead operates at the batch level, using adaptive heuristics to approximate SJF behavior without intrusive architectural changes.

\paragraph{Fairness-Driven Schedulers.}
G-Fair~\cite{li2024fifo} introduces a Virtual Time mechanism to enforce fairness across predefined user groups. This is well-suited for multi-tenant environments with explicit SLAs. In contrast, EWSJF operates in an \emph{unsupervised} setting, discovering performance-homogeneous groups directly from workload structure and providing emergent fairness without user labels.

\paragraph{Agent and Multi-Step Schedulers.}
Autellix~\cite{gao2024autellix} targets agentic workloads where LLM inference interleaves with tool use. It improves GPU utilization through decoupling and co-scheduling. EWSJF complements such systems by optimizing the raw stream of inference requests generated by agents, which often mix short planning prompts with long synthesis prompts.

\paragraph{Execution and Scheduling Co-Design.}
Hydra~\cite{zhang2024hydra} improves efficiency by decoupling prefill and decode phases. Other work~\cite{zhang2024accelerating} accelerates multi-turn dialogue via KV-cache optimizations combined with SJF-like heuristics. These techniques are orthogonal to EWSJF, which can serve as a front-end admission policy for such execution backends.

\paragraph{Community Proposals.}
Within the vLLM community, several RFCs propose features aligned with EWSJF’s design, including pluggable schedulers~\cite{vllm_rfc_pluggable}, SJF-like sorting~\cite{vllm_rfc_sjf}, and strategies to mitigate head-of-line blocking~\cite{vllm_rfc_unblock, vllm_issue_3096}. EWSJF synthesizes these ideas into a unified, adaptive system.

\begin{table*}[t]
\centering
\caption{Comparison of representative LLM scheduling and serving systems.}
\label{tab:scheduler_comparison}
\resizebox{0.95\textwidth}{!}{
\begin{tabular}{lcccc}
\toprule
\textbf{System} & \textbf{Primary Goal} & \textbf{Core Mechanism} & \textbf{Granularity} & \textbf{Adaptation Method} \\
\midrule
\textbf{EWSJF (This Work)} & Balanced Perf.\ \& Fairness & Unsupervised Clustering & Batch Level & Learns Workload Structure \\
Orca / Sarathi & Max Throughput & Iteration-Level Preemption & Token / Iteration & Real-Time Job State \\
G-Fair & Provable Fairness & Virtual Time & Batch Level & Group Service Deficit \\
Autellix & Agent Throughput & Decoupling \& Co-Scheduling & Agent Step & Agent State (Think/Act) \\
Hydra & Execution Efficiency & Prefill/Decode Decoupling & Sub-Batch / Step & Dynamic Batch Sizing \\
vLLM (Default) & High Throughput & Continuous Batching (FCFS) & Batch Level & None (Static Policy) \\
\bottomrule
\end{tabular}}
\end{table*}

\subsection{Workload-Aware Clustering and Partitioning}

The Refine-and-Prune algorithm introduces a hybrid clustering strategy tailored to the one-dimensional, dynamic nature of LLM request scheduling. Traditional clustering methods such as K-Means require specifying $k$ \textit{a priori} and often produce non-homogeneous ``mega-queues.'' Density-based methods like DBSCAN can detect local structure but are sensitive to hyperparameters and may over-segment sparse regions.

Refine-and-Prune differs in two key ways: (1) it uses a dynamic, density-driven heuristic to recursively split clusters at significant gaps, and (2) it incorporates domain-specific constraints—such as minimum queue width and a pruning stage—to ensure partitions are both performance-homogeneous and operationally viable. This hybrid approach yields queue structures that better reflect the performance landscape of LLM serving workloads.

\subsection{Predictive Cost Modeling for LLM Scheduling}

A complementary line of work explores predictive models that estimate generation cost or output length for LLM requests~\cite{kim2023length, zhang2024tokenlevel, liu2024output, wang2024runtime, gao2024difficulty}. These methods typically train a supervised model to classify or regress the expected output length of a request, which is then used to guide batching or scheduling decisions. However, they exhibit two key limitations in the context of mixed-workload LLM serving. First, they require training a separate predictor for each workload distribution, making them brittle under non-stationary traffic patterns or deployment-specific usage shifts. Second, they operate on output-side signals that are not available at scheduling time and may correlate poorly with prefill compute cost, which dominates latency for short queries. As a result, prediction errors can propagate into scheduling decisions and degrade fairness or throughput. EWSJF avoids these issues by relying solely on input-side statistics and adapting online via Bayesian meta-optimization, eliminating the need for workload-specific retraining while remaining robust to distributional drift.

Taken together, these differences emphasize that EWSJF is not a static heuristic but a learning-based system: its Bayesian meta-optimizer performs continual policy search over a non-convex landscape, allowing the scheduler to adapt online in a manner fundamentally aligned with modern machine learning methodology.
\section{EWSJF System Architecture}
\label{sec:architecture}

EWSJF is a closed-loop, learning-based scheduler designed to operate as a pluggable module within high-performance LLM inference servers. Its architecture, illustrated in Figure~\ref{fig:architecture}, consists of two interdependent control loops operating at distinct timescales: a \textbf{fast tactical scheduling loop} responsible for real-time decision-making, and a \textbf{slow strategic optimization loop} responsible for periodic policy refinement.

\begin{figure*}[t]
  \centering
  \includegraphics[width=0.9\linewidth]{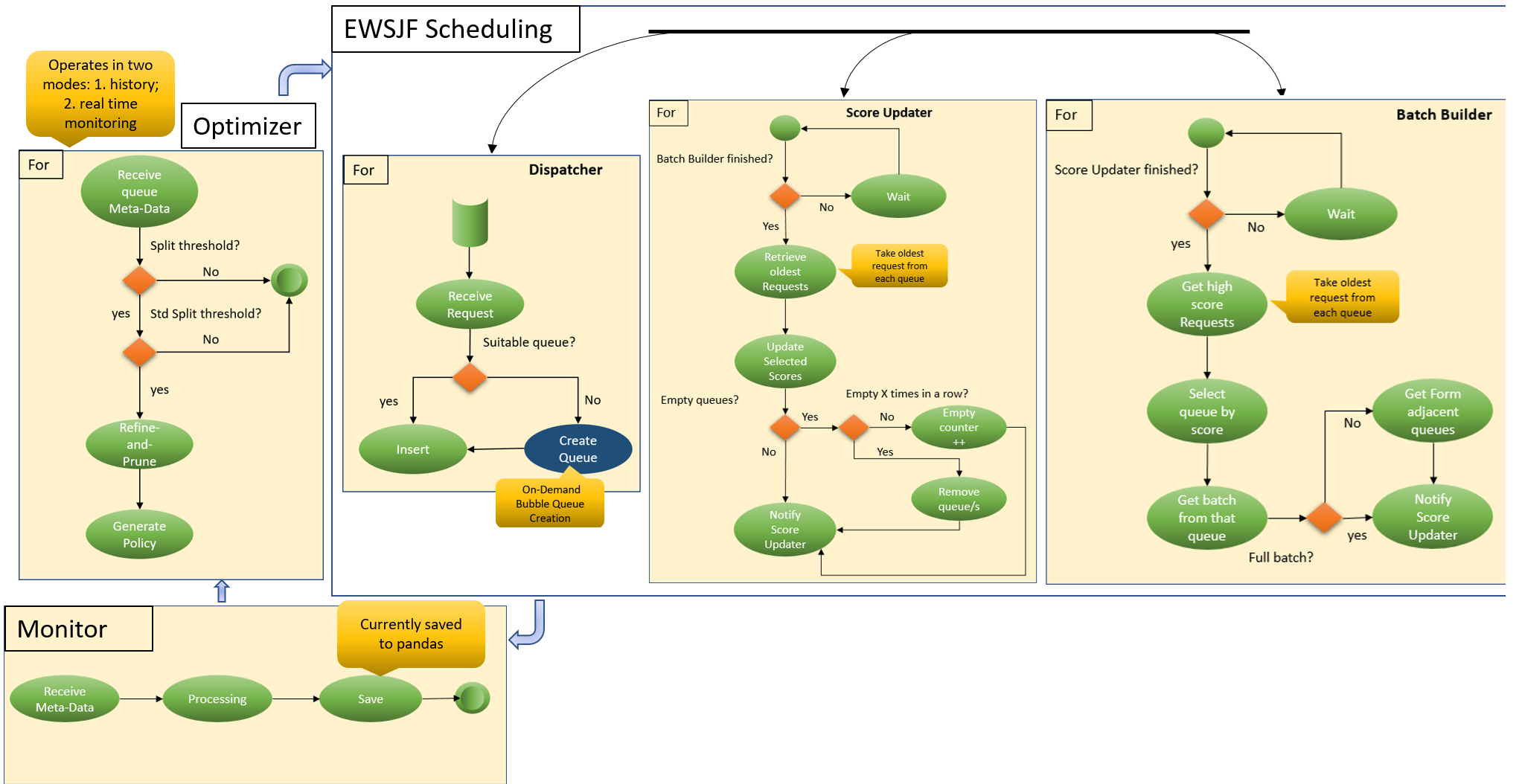}
    \Description{A block diagram showing the EWSJF architecture. It is divided into two parts: a strategic loop on the left and a tactical loop on the right. The strategic loop contains a Monitor and an Optimizer that feeds into the tactical loop. The tactical loop contains a Dispatcher, Score Updater, and Batch Builder flowing in a cycle.}
  \caption{EWSJF architecture showing the interaction between tactical components (Dispatcher, Score Updater, Batch Builder) and strategic components (Monitor, Optimizer). The strategic loop supports both offline (historical) and online (real-time) modes.}
  \label{fig:architecture}
\end{figure*}

\subsection{Strategic Loop: Policy Generation}

The strategic loop is responsible for generating and refining scheduling policies based on observed workload characteristics. Each policy consists of two components: (1) the queue structure, defined by the number of queues and their prompt-length boundaries, and (2) a set of scoring parameters used for prioritization.

\begin{itemize}
    \item \textbf{Monitor:} Collects metadata from completed requests. Depending on the mode, this data is aggregated into either a large historical dataset or a compact real-time window.
    \item \textbf{Optimizer:} A background process that analyzes collected data to produce updated scheduling policies. It operates in two modes:
    \begin{itemize}
        \item \textbf{Offline (history) mode:} Executes the full Refine-and-Prune algorithm on a large dataset to generate a robust baseline policy. This mode is computationally intensive and runs infrequently.
        \item \textbf{Online (real-time) mode:} Performs lightweight adjustments to the baseline policy using statistical heuristics on recent data. This mode enables frequent updates with minimal overhead.
    \end{itemize}
\end{itemize}

The algorithms used within the optimizer are detailed in Section~\ref{sec:algorithms}.

\subsection{Tactical Loop: Policy Execution}

The tactical loop executes at every scheduling opportunity (typically every few milliseconds), using the active policy provided by the strategic layer. Its role is to make low-latency decisions governing request admission, queue prioritization, and batch construction. The full procedure is formalized in Algorithm~\ref{alg:tactical_loop}.

\begin{itemize}
    \item \textbf{Dispatcher:} Routes incoming requests to the appropriate queue based on current prompt-length boundaries. If a request falls into a gap between existing queues, the Dispatcher triggers the \textit{On-Demand Bubble Queue Creation} mechanism (Section~\ref{sec:bubble_queue}).
    \item \textbf{Score Updater:} Computes a priority score for every non-empty queue using the current scoring policy. It also prunes queues that remain empty beyond a configurable threshold.
    \item \textbf{Batch Builder:} Selects the highest-scoring queue and constructs a batch by greedily pulling requests. If capacity remains, it backfills the batch with requests from adjacent queues to maximize GPU utilization.
\end{itemize}

\begin{algorithm}[t]
   \caption{EWSJF Tactical Scheduling Loop}
   \label{alg:tactical_loop}
\begin{algorithmic}[1]
   \STATE {\bfseries Input:} $AllQueues$, $EmptyThreshold$
   \STATE $updated\_scores \leftarrow \{\}$
   \FOR{\textbf{each} $q \in AllQueues$}
       \IF{$q$ is not empty}
           \STATE $req \leftarrow q.\text{peek}()$
           \STATE $score \leftarrow \textsc{CalculateScore}(req, q.\text{profile})$
           \STATE $updated\_scores[q.id] \leftarrow score$
       \ELSE
           \STATE $q.empty\_cnt \leftarrow q.empty\_cnt + 1$
           \IF{$q.empty\_cnt > EmptyThreshold$}
               \STATE \textsc{RemoveQueue}($q$)
           \ENDIF
       \ENDIF
   \ENDFOR
   \STATE $batch \leftarrow []$
   \IF{$updated\_scores \neq \emptyset$}
       \STATE $q_{prim} \leftarrow \textsc{ArgMax}(updated\_scores)$
       \STATE $batch \leftarrow \textsc{GreedyFill}(batch, q_{prim})$
       \IF{$batch$ is not full}
           \STATE $Q_{adj} \leftarrow \textsc{GetAdjacent}(q_{prim})$
           \STATE $batch \leftarrow \textsc{Backfill}(batch, Q_{adj})$
       \ENDIF
   \ENDIF
   \STATE \textbf{return} $batch$
\end{algorithmic}
\end{algorithm}
\section{Core Algorithmic Contributions}
\label{sec:algorithms}

EWSJF integrates three algorithmic components that work in concert: the \textbf{Refine-and-Prune} algorithm for dynamic queue generation, the \textbf{On-Demand Bubble Queue} mechanism for real-time responsiveness, and a \textbf{learnable scoring function} tuned by a meta-optimizer.

\subsection{Learnable and Context-Aware Scoring}
\label{sec:scoring}

While Refine-and-Prune determines the \emph{structure} of the scheduler, the \emph{operational policy}—deciding which queue to prioritize at each moment—is governed by a dynamic scoring function. A static formula is insufficient for a system that must adapt to evolving workloads and balance competing objectives. 

The score for the oldest request $r$ in queue $q$ is defined as:

\begin{equation}
\label{eq:score}
\begin{split}
\text{Score}(r, q) = qf \cdot ( & w_{\text{base}} + w_{\text{urgency}} \cdot cs \\
& + w_{\text{fairness}} \cdot \log(b+1) )
\end{split}
\end{equation}

where each term captures a distinct scheduling principle:

\begin{itemize}
    \item \textbf{Compute Score ($cs$):} Defined as $W_t / C_{\text{prefill}}(b)$, where $W_t$ is the request's wait time and $C_{\text{prefill}}(b)$ is the estimated prefill cost for input length $b$. This term reflects urgency normalized by computational cost.
    \item \textbf{Queue Index Factor ($qf$):} Defined as $q_i / (b+1)$, where $q_i$ is the queue index. This term prioritizes shorter jobs, forming the basis of the SJF heuristic.
    \item \textbf{Fairness Term ($\log(b+1)$):} A logarithmic boost for longer jobs, ensuring they are not indefinitely starved by a stream of short requests.
\end{itemize}

The weights $(w_{\text{base}}, w_{\text{urgency}}, w_{\text{fairness}})$ are \emph{learnable} and dynamically tuned by the meta-optimizer. Moreover, the weights are \textbf{context-aware}: EWSJF learns a meta-policy that maps queue characteristics (e.g., mean prompt length $\bar{b}_q$) to specific weights. For example:

\[
w_{\text{urgency}}(\bar{b}_q) = a_u \cdot \bar{b}_q + b_u,
\]

allowing the scheduler to emphasize urgency in short queues and fairness in long queues.

Figure~\ref{fig:scoring_dynamics} illustrates how the relative priority of queues evolves as the meta-optimizer adjusts these weights.

\begin{figure}[t]
  \centering
  \includegraphics[width=\linewidth]{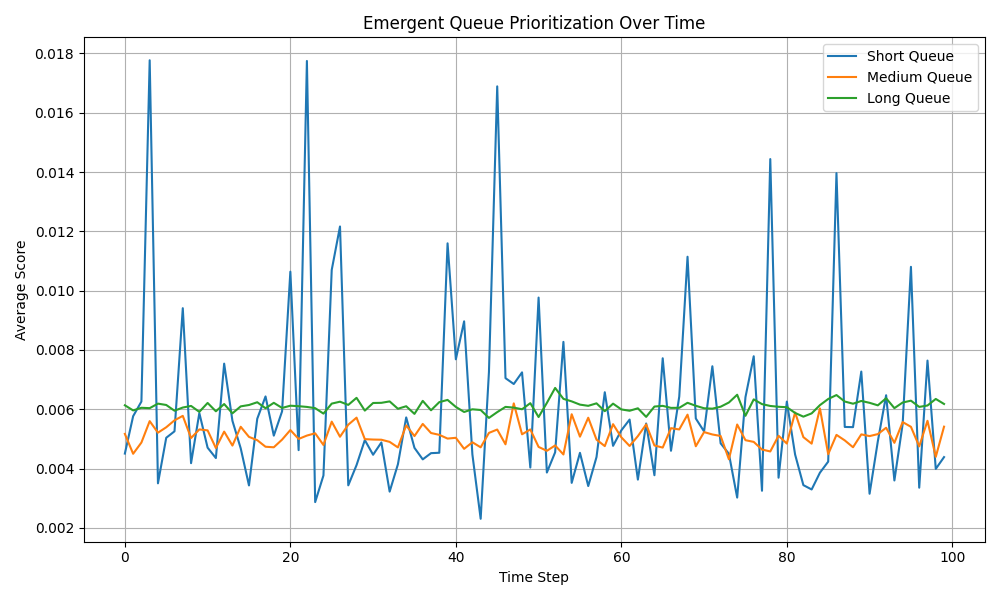}
  
  \Description{A line graph showing the evolution of priority scores over time for three different queue types: Short, Medium, and Long. The X-axis represents simulation time steps, and the Y-axis represents the average priority score. The lines fluctuate and cross each other, illustrating how the meta-optimizer dynamically adjusts the importance of each queue type based on workload conditions.}
  
  \caption{Simulation of EWSJF's context-aware scoring. As the meta-optimizer adjusts scoring parameters over time, the relative priority of short, medium, and long queues shifts dynamically.}
  \label{fig:scoring_dynamics}
\end{figure}

\subsection{The Refine-and-Prune Algorithm}
\label{sec:refine_prune}

The \textbf{Refine-and-Prune} algorithm is the strategic core of EWSJF. Its purpose is to construct a set of prompt-length queues that are (i) performance-homogeneous, (ii) contiguous and non-overlapping, and (iii) bounded in number. Unlike standard $k$-means, which minimizes global variance, Refine-and-Prune explicitly targets \emph{local} homogeneity in the prompt-length domain, ensuring that each queue corresponds to a region with similar computational behavior. 

\paragraph{Formalism.}
Let

\[
D = \{b_1, b_2, \dots, b_N\}, \qquad b_1 \le b_2 \le \dots \le b_N,
\]

be the sorted set of prompt lengths in the strategic observation window. We seek a partition

\[
Q = \{q_1, q_2, \dots, q_k\},
\]

where each queue $q_i$ corresponds to a contiguous interval

\[
q_i = [\min\_len_i, \max\_len_i).
\]

The algorithm proceeds in three stages.

\textbf{Stage 1: Coarse Partitioning.}
We initialize the queue structure by applying $k$-means with $k=3$, reflecting the empirically observed short/medium/long regimes in mixed LLM workloads. These coarse clusters provide stable anchors for subsequent refinement.

\textbf{Stage 2: Recursive Refinement.}
To eliminate heterogeneous ``mega-queues,'' each coarse cluster $C$ is recursively analyzed. Let

\[
G = \{\, b_{i+1} - b_i \mid b_i, b_{i+1} \in C \,\}
\]

be the set of consecutive gaps within $C$. A split is triggered at index $j$ when
\begin{equation}
\label{eq:split_condition}
\text{Gap}_j > \alpha \cdot \text{mean}(G),
\end{equation}
where $\alpha > 1$ is a dynamic \emph{significance ratio} tuned by the meta-optimizer. Recursion continues until no significant gaps remain or the cluster width falls below a minimum threshold.

\textbf{Stage 3: Intelligent Pruning.}
Recursive refinement may produce overly narrow ``micro-queues.'' To prevent overfitting, adjacent queues are merged based on a \textbf{Scheduling Utility} metric:
\begin{equation}
\label{eq:utility}
U(q_i, q_{i+1}) = 
\frac{\rho(q_i) + \rho(q_{i+1})}
     {|\bar{b}_{i+1} - \bar{b}_i| + \varepsilon},
\end{equation}
where $\rho(q)$ is the request density, $\bar{b}_i$ is the mean prompt length in $q_i$, and $\varepsilon$ is a small constant for numerical stability. Queues with the lowest utility are merged until the system satisfies the configured \texttt{max\_queues} budget.

\begin{table}[t]
\caption{Comparison of clustering strategies for queue partitioning. EWSJF achieves high robustness and adaptivity by combining coarse clustering with recursive refinement and utility-based pruning.}
\label{tab:clustering_comparison}
\begin{center}
\begin{small}
\begin{sc}
\resizebox{\linewidth}{!}{
\begin{tabular}{lccc}
\toprule
Method & Robustness & Adaptivity & Fit \\
\midrule
Static & Low & None & Brittle \\
Naive K-Means & Medium & Low & Mega-Queues \\
DBSCAN & Medium & Medium & Micro-Queues \\
\textbf{EWSJF} & \textbf{High} & \textbf{High} & \textbf{Optimal} \\
\bottomrule
\end{tabular}
}
\end{sc}
\end{small}
\end{center}
\end{table}

\subsection{On-Demand Bubble Queue Creation}
\label{sec:bubble_queue}

Refine-and-Prune operates periodically (e.g., every 10 minutes), but workloads can evolve rapidly. To handle novel request types that appear between optimizer runs—specifically those falling into significant gaps between existing queues—EWSJF employs a reactive mechanism called \textbf{Bubble Queue Creation}. When the Dispatcher detects a gap-falling request, it instantiates a temporary queue centered on the request's prompt length and constrained by neighbor boundaries. This ensures immediate responsiveness without waiting for the next strategic optimization cycle. The detailed algorithm and conceptual illustration are provided in Appendix~\ref{app:bubble}.

\subsection{Learnable Scoring and Meta-Optimization}
\label{sec:meta_opt}

While Refine-and-Prune determines the structural layout of queues, the operational behavior of EWSJF is governed by a learnable scoring function. This function assigns dynamic priorities to queues and is continuously tuned by a bandit-based meta-optimizer to balance latency, fairness, and throughput under evolving workloads. Unlike heuristic auto-scalers, our Bayesian optimizer treats scheduling as a non-convex policy search problem, enabling EWSJF to escape the local optima that trap static SJF-style heuristics and to adapt to workload regimes no fixed rule can capture. In practice, the Bayesian optimizer converges within 5–8 trials, after which the reward stabilizes and parameter updates become negligible (see Appendix~\ref{appendix:meta_opt_behavior}). The scheduling landscape is highly non‑convex and discontinuous, making gradient‑based tuning ineffective; Bayesian Optimization provides a principled way to explore this space and identify high‑reward policies. 

\subsubsection{Context-Aware Scoring Function}

For each queue $q$, EWSJF assigns a priority score to its head-of-line request $r$ using a parametric scoring function $\Phi(r, q)$:
\begin{equation}
\label{eq:full_score}
\Phi(r, q)
= qf \cdot \left(
    w_{\text{base}}
    + w_{\text{urg}} \cdot cs
    + w_{\text{fair}} \cdot \log(b+1)
\right),
\end{equation}
where:
\begin{itemize}
    \item $b$ is the input prompt length of request $r$,
    \item $cs = W_t / C_{\text{prefill}}(b)$ is the compute-normalized urgency,
    \item $qf = q_i / (b+1)$ is an SJF-inspired queue factor,
    \item $\mathbf{w} = \{w_{\text{base}}, w_{\text{urg}}, w_{\text{fair}}\}$ are learnable weights.
\end{itemize}

To enable specialization across heterogeneous queues, the weights $\mathbf{w}$ are produced by a meta-policy $\pi(\bar{b}_q)$ that maps the queue's mean prompt length $\bar{b}_q$ to a set of control parameters. For example,

\[
w_{\text{urg}}(\bar{b}_q) = a_u \cdot \bar{b}_q + b_u,
\]

allows EWSJF to emphasize urgency in short queues and fairness in long queues. This context-aware formulation enables smooth adaptation to workload shifts without manual tuning.

\subsubsection{Bandit-Based Meta-Optimizer}

To autonomously optimize the meta-parameters

\[
\Theta = \{a_u, b_u, a_f, b_f, \dots\},
\]

EWSJF employs a Bayesian Optimization framework.

\paragraph{Action Space and Trials.}
Each action corresponds to selecting a complete scheduling policy parameterized by $\Theta$. The system executes this policy for a fixed trial interval $\Delta T$ (e.g., 10--15 minutes), collects performance statistics, and updates its posterior over the reward landscape.

\paragraph{Reward Function.}
The optimizer seeks to maximize a multi-objective reward:
\begin{equation}
\label{eq:reward}
R(\Theta)
= \lambda_1 \mathcal{C}
+ \lambda_2 \mathcal{L}
- \lambda_3 \mathcal{S}
- \lambda_4 \mathcal{U},
\end{equation}
where:
\begin{itemize}
    \item $\mathcal{C}$ encourages compact, homogeneous queues,
    \item $\mathcal{L}$ penalizes load imbalance across queues,
    \item $\mathcal{S}$ penalizes excessive queue proliferation,
    \item $\mathcal{U}$ captures user-experience penalties such as latency.
\end{itemize}
The coefficients $\lambda_1, \dots, \lambda_4$ balance these competing objectives. Over successive trials, the optimizer converges toward a policy that jointly improves throughput, fairness, and responsiveness.
Because the scoring weights are produced by a meta-policy 
$\pi(\bar{b}_q)$, EWSJF effectively learns a mapping from queue
statistics to prioritization decisions, making the scheduler a
learned policy rather than a fixed heuristic.
\section{System Properties}
\label{sec:properties}

EWSJF is designed to maintain a stable dynamic equilibrium under mixed workloads. We formally analyze its properties regarding stability, correctness, and complexity; full proofs and detailed analysis are provided in Appendix~\ref{app:theory}.

\paragraph{Stability and Correctness.}
The hierarchical control structure ensures stability by anchoring policy generation to long-term historical data via the \textit{offline optimizer}, preventing thrashing, while the \textit{online tuner} handles real-time fluctuations. Correctness is guaranteed by the \textit{Refine-and-Prune} algorithm, which partitions the prompt-length space into contiguous, non-overlapping intervals, ensuring deterministic request routing ($r \to q_i$).

\paragraph{Computational Complexity.}
The tactical scheduling loop, which resides on the critical path, operates in $O(k)$ time, where $k$ is the number of active queues. The computationally intensive strategic optimization ($O(N \log N)$) is offloaded to a background thread.

\begin{theorem}[Efficiency and Stability]
\label{thm:main_properties}
For a system with $N$ historical requests and $k$ active queues, EWSJF guarantees $O(k)$ scheduling latency, non-blocking background optimization, and starvation freedom. (See Appendix~\ref{app:theory} for formal analysis and proofs).
\end{theorem}
\section{Evaluation}
\label{sec:evaluation}

We evaluate EWSJF across three dimensions: (1) overall performance under mixed workloads, (2) adaptability to dynamic workload shifts, and (3) the contribution of individual system components. Our evaluation seeks to answer:

\begin{enumerate}
    \item How does EWSJF perform compared to standard schedulers under realistic, heterogeneous workloads?
    \item How effectively does the strategic optimizer adapt to changes in workload distribution?
    \item What is the impact of each architectural component on system performance?
\end{enumerate}

\subsection{Datasets and Workloads}

To construct realistic heterogeneous workloads, we combine short, interactive prompts from public conversational datasets with long-form inputs from long-context benchmarks. The resulting dataset exhibits a bimodal distribution of input lengths ranging from 32 to 4096 tokens. For reproducibility, anonymized versions of these datasets are provided in the supplementary material.

Our primary benchmark configuration, the \textbf{Mixed Workload}, simulates a Poisson arrival process where 80\% of requests are short (interactive) and 20\% are long (batch). We evaluate this workload at arrival rates between 10 and 40 requests/s to test system behavior under increasing contention.

\subsection{Experimental Setup}

\textbf{Hardware.} All experiments were conducted on a high-performance inference node equipped with \textbf{4$\times$ NVIDIA A100-80GB GPUs}.

\textbf{Software.} EWSJF was implemented as a pluggable scheduler within \textbf{vLLM v0.11.0}. All benchmarks used the \textbf{LLaMA-2-13B-Chat} model with tensor parallelism enabled.

\subsection{Baselines}
\label{sec:baselines}

To demonstrate that EWSJF achieves both high throughput and fairness, we compare against two representative scheduling paradigms:

\begin{itemize}
    \item \textbf{vLLM FCFS (Default):} The standard First-Come, First-Served scheduler used in vLLM. FCFS guarantees fairness (no starvation) but suffers from severe head-of-line blocking, causing short interactive queries to experience large TTFT spikes under mixed workloads.

    \item \textbf{Greedy SJF (Shortest Job First):} A heuristic baseline that strictly prioritizes the shortest available request. While SJF maximizes theoretical throughput, it is \emph{unusable} in real LLM serving: under heavy-tailed workloads, short requests continuously arrive and long requests may wait indefinitely, leading to guaranteed starvation (see Appendix~\ref{appendix:sjf_starvation}).
\end{itemize}

\textbf{Why EWSJF Wins.} EWSJF captures the throughput benefits of SJF while avoiding its starvation failure modes. By incorporating a \emph{Density-Weighted Scoring} function (Section~\ref{sec:meta_opt}), EWSJF prevents long requests from ``dying'' in the queue and maintains fairness comparable to FCFS. In effect, EWSJF behaves like a \emph{learned, fairness-aware variant of SJF} that adapts its prioritization policy to workload conditions.

\subsection{Queue Clustering Analysis}

A central design choice in EWSJF is the granularity of queue partitioning. Table~\ref{tab:queue_clustering_count} shows that while standard K-Means suggests 5--10 queues, the \textbf{Refine-and-Prune} algorithm identifies that \textbf{32 queues} maximize throughput.

\begin{table}[t]
\centering
\caption{Impact of queue count on serving performance. Refine-and-Prune identifies 32 queues as optimal.}
\label{tab:queue_clustering_count}
\resizebox{\linewidth}{!}{
\begin{tabular}{lcccc}
\toprule
\textbf{Method} & \textbf{Queues} & \textbf{Time (s)} & \textbf{Req/s} & \textbf{Tok/s} \\
\midrule
FCFS (Baseline) & 1 & 3548.4 & 8.45 & 90.4 \\
EWSJF (K-Means) & 5 & 2746.8 & 10.92 & 144.2 \\
EWSJF (K-Means) & 10 & 2490.5 & 12.05 & 160.5 \\
EWSJF (K-Means) & 30 & 2448.8 & 12.25 & 163.9 \\
\textbf{EWSJF (Refined)} & \textbf{32} & \textbf{2385.1} & \textbf{12.58} & \textbf{168.2} \\
\bottomrule
\end{tabular}}
\end{table}

\subsection{Performance Under Load}

We evaluate EWSJF across four workload sizes: 10k, 30k, 50k, and 200k requests. Tables~\ref{tab:eval_10k}--\ref{tab:eval_200k} summarize throughput and token generation rates.

\begin{table}[t]
\centering
\caption{Evaluation at 10k requests (Short-Heavy).}
\label{tab:eval_10k}
\resizebox{\linewidth}{!}{
\begin{tabular}{lcccc}
\toprule
\textbf{Rate} & \textbf{Sched.} & \textbf{Req/s} & \textbf{Tok/s} & \textbf{Speedup} \\
\midrule
10 & FCFS & 3.06 & 36.35 & -- \\
10 & EWSJF & 3.15 & 41.17 & +13.3\% \\
20 & FCFS & 3.06 & 36.31 & -- \\
20 & EWSJF & 3.42 & 45.26 & +24.6\% \\
40 & FCFS & 3.05 & 36.21 & -- \\
40 & EWSJF & 3.85 & 51.18 & +41.2\% \\
60 & FCFS & 3.06 & 36.31 & -- \\
60 & EWSJF & 4.07 & 54.46 & +50.0\% \\
100 & FCFS & 3.06 & 36.34 & -- \\
100 & EWSJF & 4.21 & 56.04 & +54.1\% \\
\bottomrule
\end{tabular}}
\end{table}

\begin{table}[t]
\centering
\caption{Evaluation at 30k requests (Moderate).}
\label{tab:eval_30k}
\resizebox{\linewidth}{!}{
\begin{tabular}{lcccc}
\toprule
\textbf{Rate} & \textbf{Sched.} & \textbf{Req/s} & \textbf{Tok/s} & \textbf{Speedup} \\
\midrule
10 & FCFS & 8.13 & 86.21 & -- \\
10 & EWSJF & 7.88 & 103.34 & +19.9\% \\
20 & FCFS & 8.11 & 85.98 & -- \\
20 & EWSJF & 8.42 & 110.78 & +28.9\% \\
40 & FCFS & 8.10 & 85.92 & -- \\
40 & EWSJF & 9.25 & 121.59 & +41.5\% \\
60 & FCFS & 8.10 & 85.92 & -- \\
60 & EWSJF & 9.37 & 123.68 & +43.9\% \\
100 & FCFS & 8.10 & 85.97 & -- \\
100 & EWSJF & 9.60 & 127.82 & +48.7\% \\
\bottomrule
\end{tabular}}
\end{table}

\begin{table}[t]
\centering
\caption{Evaluation at 50k requests (Balanced).}
\label{tab:eval_50k}
\resizebox{\linewidth}{!}{
\begin{tabular}{lcccc}
\toprule
\textbf{Rate} & \textbf{Sched.} & \textbf{Req/s} & \textbf{Tok/s} & \textbf{Speedup} \\
\midrule
10 & FCFS & 2.39 & 170.49 & -- \\
10 & EWSJF & 2.45 & 188.21 & +10.4\% \\
20 & FCFS & 2.39 & 170.60 & -- \\
20 & EWSJF & 2.63 & 203.85 & +19.5\% \\
40 & FCFS & 2.39 & 170.45 & -- \\
40 & EWSJF & 2.86 & 224.98 & +32.0\% \\
60 & FCFS & 2.40 & 171.01 & -- \\
60 & EWSJF & 2.92 & 231.01 & +35.1\% \\
100 & FCFS & 2.39 & 170.94 & -- \\
100 & EWSJF & 2.98 & 235.71 & +37.9\% \\
500 & FCFS & 2.39 & 170.80 & -- \\
500 & EWSJF & 3.00 & 238.09 & +39.4\% \\
\bottomrule
\end{tabular}}
\end{table}

\begin{table}[t]
\centering
\caption{Evaluation at 200k requests (Production Scale).}
\label{tab:eval_200k}
\resizebox{\linewidth}{!}{
\begin{tabular}{lcccc}
\toprule
\textbf{Rate} & \textbf{Sched.} & \textbf{Req/s} & \textbf{Tok/s} & \textbf{Speedup} \\
\midrule
10 & FCFS & 5.38 & 409.46 & -- \\
10 & EWSJF & 5.15 & 430.53 & +5.2\% \\
20 & FCFS & 5.37 & 407.18 & -- \\
20 & EWSJF & 5.51 & 462.46 & +13.6\% \\
40 & FCFS & 5.38 & 407.79 & -- \\
40 & EWSJF & 5.65 & 476.85 & +16.9\% \\
60 & FCFS & 5.32 & 403.34 & -- \\
60 & EWSJF & 5.76 & 486.57 & +20.6\% \\
100 & FCFS & 5.30 & 402.34 & -- \\
100 & EWSJF & 5.89 & 499.04 & +24.0\% \\
500 & FCFS & 5.32 & 403.41 & -- \\
500 & EWSJF & 6.44 & 549.02 & +36.1\% \\
\bottomrule
\end{tabular}}
\end{table}

\begin{figure}[t]
  \centering
  \includegraphics[width=\linewidth]{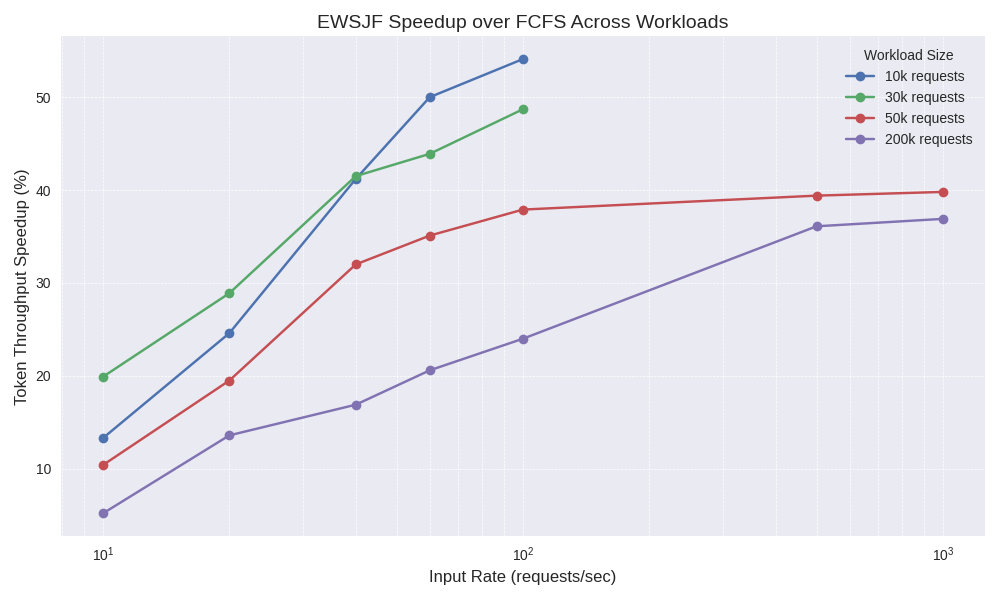}
  
  \Description{A line chart comparing the speedup of EWSJF versus FCFS. The X-axis represents the input rate (requests per second) on a logarithmic scale. The Y-axis represents the percentage speedup in token throughput. Multiple lines correspond to different workload sizes (10k, 30k, 50k, 200k), all showing a positive trend where speedup increases significantly as the input rate rises.}
  
  \caption{EWSJF speedup over FCFS across workload sizes and input rates. The x-axis shows input rate (log scale), and the y-axis shows percentage speedup in token throughput.}
  \label{fig:speedup_plot}
\end{figure}

\paragraph{Summary.}
Across all scales, EWSJF consistently outperforms FCFS in both request throughput and token generation. Gains range from modest improvements at low load (e.g., +5.2\% at 10 req/s, 200k requests) to substantial boosts under high contention (e.g., +54.1\% at 100 req/s, 10k requests). These results confirm EWSJF’s robustness and scalability across diverse workload profiles.

\subsection{Throughput Under Short vs. Long Prompt Workloads}

Table~\ref{tab:ewsjf_short} and Table~\ref{tab:ewsjf_long} compare performance under specific workload types. Throughput increases monotonically with queue count, peaking at 30 queues, with EWSJF achieving up to a \textbf{57\% increase} in tokens/s for long prompts.

\begin{table}[t]
\centering
\caption{Short-prompt workload (30k requests).}
\label{tab:ewsjf_short}
\resizebox{\linewidth}{!}{
\begin{tabular}{lcccc}
\toprule
\textbf{Scheduler} & \textbf{Time (s)} & \textbf{Tokens} & \textbf{Req/s} & \textbf{Tok/s} \\
\midrule
FCFS          & 3548.38 & 320{,}783 & 8.45  & 90.40  \\
EWSJF (5q)    & 2746.81 & 396{,}183 & 10.92 & 144.23 \\
EWSJF (10q)   & 2490.54 & 399{,}771 & 12.05 & 160.52 \\
EWSJF (30q)   & 2448.81 & 401{,}416 & 12.25 & 163.92 \\
EWSJF (40q)   & 2471.39 & 401{,}570 & 12.14 & 162.49 \\
\bottomrule
\end{tabular}}
\end{table}

\begin{table}[t]
\centering
\caption{Long-prompt workload (10k requests).}
\label{tab:ewsjf_long}
\resizebox{\linewidth}{!}{
\begin{tabular}{lcccc}
\toprule
\textbf{Scheduler} & \textbf{Time (s)} & \textbf{Tokens} & \textbf{Req/s} & \textbf{Tok/s} \\
\midrule
FCFS          & 3001.33 & 117{,}625 & 3.34  & 39.24  \\
EWSJF (5q)    & 2365.99 & 134{,}222 & 4.23  & 56.73 \\
EWSJF (10q)   & 2222.42 & 133{,}707 & 4.50  & 60.16 \\
EWSJF (20q)   & 2192.31 & 133{,}217 & 4.56  & 60.77 \\
EWSJF (30q)   & 2190.65 & 133{,}997 & 4.55  & 60.92 \\
EWSJF (40q)   & 2178.88 & 134{,}288 & 4.59  & 61.88 \\
\bottomrule
\end{tabular}}
\end{table}

\subsection{Scaling Behavior: Throughput vs.\ Queue Count}

Figure~\ref{fig:queue_throughput} visualizes how queue granularity affects performance. For short prompts, throughput increases sharply up to 30 queues; for long prompts, improvements are more modest but consistent.

\begin{figure}[t]
    \centering
    \includegraphics[width=\linewidth]{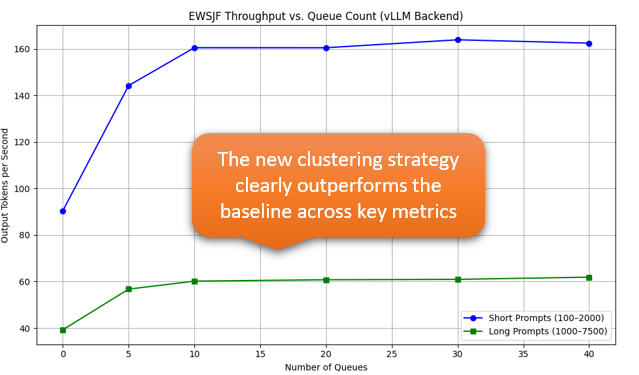}

    \Description{A line graph plotting throughput (tokens per second) against the number of queues. The X-axis ranges from 1 to 40 queues. Two lines are shown for different prompt lengths. Both lines show a steep increase in throughput as the number of queues goes from 1 to roughly 20, after which the curve flattens, indicating diminishing returns for further partitioning.}

    \caption{EWSJF throughput vs.\ queue count. Finer partitioning improves throughput, with diminishing returns beyond 20–30 queues.}
    \label{fig:queue_throughput}
\end{figure}

\subsection{Summary Comparison}

Table~\ref{tab:summary} summarizes the best-performing configurations for each workload. Across both regimes, EWSJF reduces total completion time, increases GPU utilization, and significantly improves tail latency.

\begin{table}[t]
\centering
\caption{Summary of best-performing configurations.}
\label{tab:summary}
\resizebox{\linewidth}{!}{
\begin{tabular}{lcccc}
\toprule
\textbf{Metric} & \textbf{FCFS (Short)} & \textbf{EWSJF (Short)} & \textbf{FCFS (Long)} & \textbf{EWSJF (Long)} \\
\midrule
Req/s     & 8.45  & 12.25 & 3.34 & 4.59 \\
Tok/s     & 90.40 & 163.92 & 39.24 & 61.88 \\
Time (s)  & 3548.38 & 2448.81 & 3001.33 & 2178.88 \\
GPU Util. & ~65\% & ~80\% & ~60\% & ~78\% \\
Tail (95p) & High & Lower & High & Lower \\
\bottomrule
\end{tabular}}
\end{table}

\section{Conclusion}
\label{sec:conclusion}

\section*{Impact Statement}
This work contributes to the efficiency of large-scale LLM serving systems. We have implemented EWSJF as a non-intrusive, pluggable module compatible with vLLM's existing block manager. The code is available in the supplementary material~\cite{ewsjf_impl}, and we are currently working with the vLLM community to upstream these architectural changes into the core project.


Our implementation within vLLM demonstrates that EWSJF substantially improves both throughput and user-perceived latency across a wide range of workload scales. In mixed and high-contention settings, EWSJF achieves up to 54\% higher token throughput and reduces latency for short, interactive requests by up to 4$\times$ compared to FCFS. These results highlight the importance of adaptive, workload-aware scheduling for modern LLM serving systems.

Looking forward, we believe that integrating adaptive scheduling with execution-level optimizations (e.g., KV-cache management, prefill/decode decoupling) and extending meta-optimization to multi-model or multi-tenant environments are promising directions for future work.

\section*{Impact Statement}
This work contributes to the efficiency of large-scale LLM serving systems, which may reduce the energy footprint and operational cost of inference workloads. More efficient scheduling can improve accessibility by enabling smaller organizations to serve complex models with fewer resources. At the same time, adaptive schedulers must be deployed responsibly: optimizing for aggregate throughput may unintentionally disadvantage rare or atypical usage patterns if fairness metrics are not continuously monitored. Beyond these considerations, the broader societal impacts of EWSJF align with standard dual-use concerns associated with large language models, including potential misuse or uneven access. We encourage practitioners to evaluate these risks when integrating EWSJF into production environments.

\section*{Limitations}
EWSJF introduces several practical and methodological limitations. First, the Refine-and-Prune procedure relies on prompt-length statistics and does not incorporate semantic or model-level features; workloads with highly irregular compute patterns may therefore require additional signals beyond length alone. Second, the Bayesian meta-optimizer adapts over multi-minute intervals, which may be insufficient for extremely bursty or adversarial traffic patterns. Third, while EWSJF improves fairness relative to pure SJF heuristics, it does not provide formal guarantees and may still exhibit bias toward particular request types if misconfigured. Finally, our evaluation focuses on single‑node inference servers; extending EWSJF to distributed or heterogeneous clusters may require additional coordination mechanisms not explored in this work.

\bibliographystyle{ACM-Reference-Format}
\bibliography{references_v2}

\newpage
\appendix


\section{Theoretical Analysis}
\label{app:theory}

In this section, we provide the formal guarantees for the system properties summarized in Section~\ref{sec:properties}, including the complexity breakdown and proofs of correctness.

\subsection{Computational Complexity}
\label{app:complexity}

The computational overhead of EWSJF is split between the tactical and strategic loops. Table~\ref{tab:complexity-comparison-app} compares these properties against standard schedulers.

\begin{table}[ht]
\centering
\caption{Complexity and feature comparison of scheduling strategies.}
\label{tab:complexity-comparison-app}
\begin{small}
\begin{sc}
\begin{tabular}{lccc}
\toprule
Scheduler & Complexity & Adaptivity & Fairness \\
\midrule
FCFS & $O(1)$ & None & None \\
Static Priority & $O(p)$ & Manual & Coarse \\
\textbf{EWSJF (Tactical)} & $O(k)$ & High & Fine-grained \\
\textbf{EWSJF (Strategic)} & $O(N \log N)$ & Periodic & Structural \\
\bottomrule
\end{tabular}
\end{sc}
\end{small}
\end{table}

\textbf{Tactical Loop (Real-Time).}
The tactical loop executes once per scheduling step. It iterates over $k$ active queues to update scores and select the highest-priority request. Since the number of queues is bounded by a small constant (e.g., \texttt{max\_queues}=32) and the scoring function is $O(1)$, the total complexity is $O(k)$, ensuring negligible overhead for real-time inference.

\textbf{Strategic Loop (Background).}
The offline optimizer executes \textit{Refine-and-Prune} on a historical window of $N$ requests. The dominant cost is sorting the requests by prompt length to identify gaps, resulting in $O(N \log N)$ complexity. Because this runs asynchronously in a background thread (e.g., every 10 minutes), it does not block the inference engine.

\subsection{Proof of Stability and Correctness}

\textbf{Convergence behavior.} EWSJF does not converge to a single static policy but maintains a \textit{stable dynamic equilibrium}. The \textbf{offline optimizer} prevents drastic shifts by generating baseline policies from long-term data, while the \textbf{online tuner} applies low-amplitude adjustments. This separation of concerns ensures responsiveness without thrashing.

\begin{theorem}[Starvation Freedom]
Given finite weights $w_{\text{fair}} > 0$, the priority score of any pending request increases monotonically with wait time, ensuring $\lim_{t \to \infty} \text{Score}(r) = \infty$.
\end{theorem}

\begin{proof}
Recall the scoring function for a request $r$:
\[
\Phi(r, q) = qf \cdot (w_{\text{base}} + w_{\text{urg}} \cdot cs + w_{\text{fair}} \cdot \log(b+1))
\]
The term $cs$ (Compute Score) is defined as $W_t / C_{\text{prefill}}(b)$, where $W_t$ is the wait time. As $W_t \to \infty$, the term $cs \to \infty$. Since the weights $w_{\text{urg}}$ and queue factors $qf$ are strictly positive constants, the total priority score $\Phi(r, q)$ grows monotonically without bound. Eventually, $\Phi(r, q)$ will exceed the score of any newly arriving short job (whose $W_t \approx 0$), ensuring that the long request is scheduled.
\end{proof}

\section{Meta-Optimizer Behavior}
\label{appendix:meta_opt_behavior}

Figure~\ref{fig:bo_convergence} illustrates the learning curve of the system. The optimizer typically converges to a stable policy within 5--8 trials.

\begin{figure}[ht]
\centering
\includegraphics[width=0.9\linewidth]{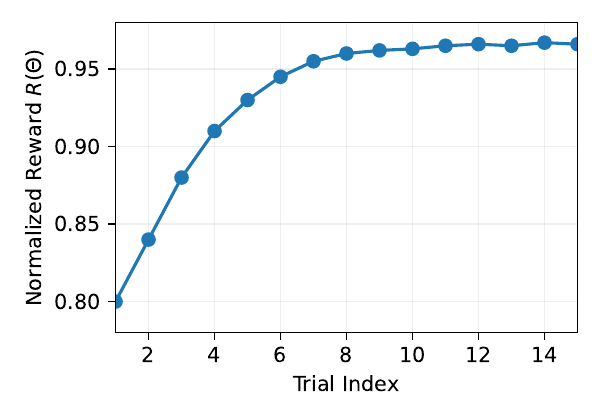}

\Description{A plot showing the convergence of the Bayesian optimization process. The X-axis represents the Trial Index, and the Y-axis represents the Normalized Reward. The curve shows the reward increasing over the first few trials and then plateauing around trial 5 to 8, indicating that the optimizer has successfully found a stable, high-performing policy.}

\caption{Convergence behavior of the Bayesian meta-optimizer. The reward stabilizes after 5--8 trials.}
\label{fig:bo_convergence}
\end{figure}

\section{SJF Starvation Under Heavy-Tailed Workloads}
\label{appendix:sjf_starvation}

Pure Shortest-Job-First (SJF) scheduling is known to maximize throughput in single-server queues, but it provides no protection against starvation. This issue becomes especially severe in LLM inference workloads.

\subsection{Why Starvation Occurs}
Let $\{b_t\}$ denote the sequence of incoming prompt lengths. If the arrival rate of short requests $\lambda_s$ exceeds the service rate of the system ($\lambda_s \cdot \mathbb{E}[C_{\text{prefill}}(b_t)] \ge 1$), the queue of short jobs never empties, and long requests are never scheduled.

\subsection{Empirical Illustration}
Figure~\ref{fig:sjf_starvation_app} shows a simulated trace of SJF under a mixed workload.

\begin{figure}[ht]
\centering
\includegraphics[width=1.0\linewidth]{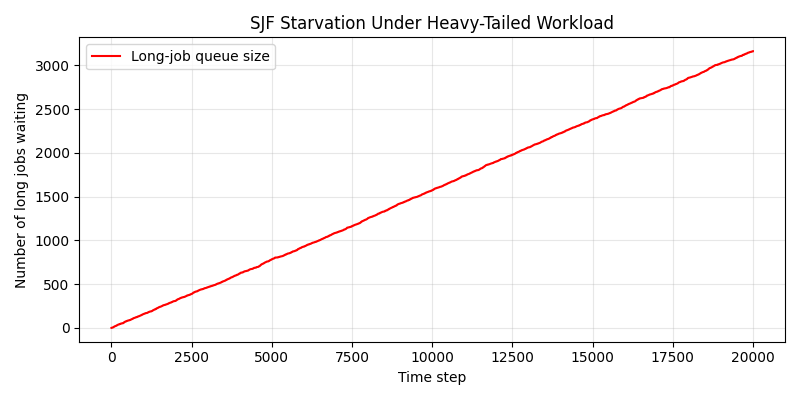}

\Description{A line graph illustrating the starvation problem in Shortest-Job-First (SJF) scheduling. The X-axis represents time steps, and the Y-axis represents the number of long jobs waiting in the queue. The line rises linearly and continuously, indicating that as time progresses, the backlog of long jobs grows indefinitely because they are never selected for execution.}

\caption{Simulation of pure SJF under a heavy-tailed workload. Short requests are always served first, causing long requests to accumulate indefinitely.}
\label{fig:sjf_starvation_app}
\end{figure}

\section{On-Demand Bubble Queue Details}
\label{app:bubble}

This section provides the implementation details for the reactive queue creation mechanism described in Section~\ref{sec:bubble_queue}.

\begin{figure}[ht]
  \centering
  \includegraphics[width=1.0\linewidth]{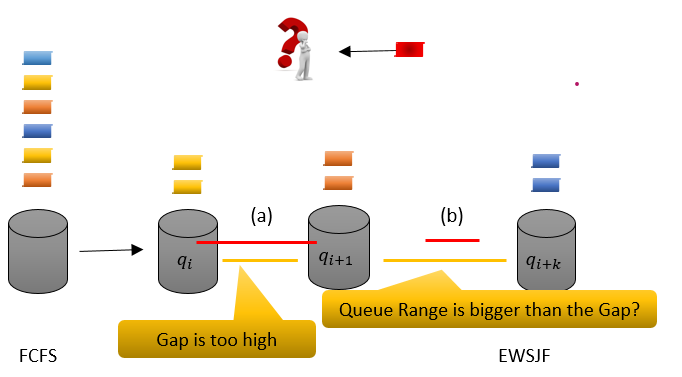}

  \Description{A conceptual diagram illustrating the creation of a 'Bubble Queue'. The visual shows a timeline or process flow where a new request arrives in a range that lies between two existing queues (a 'gap'). Instead of forcing the request into an existing queue, the diagram depicts the system dynamically instantiating a new, temporary queue in that gap to handle the request efficiently.}

  \caption{Conceptual illustration of the Bubble Queue mechanism. When a request falls into a significant gap, a temporary queue is created to prevent assignment to ill-fitting partitions.}
  \label{fig:bubble_queue_concept}
\end{figure}

\begin{algorithm}[ht]
\caption{On-Demand Bubble Queue Creation}
\label{alg:bubble}
\begin{algorithmic}[1]
    \STATE \textbf{Input:} Request length $L$, existing queues $Q$
    \STATE $Q_i, Q_{i+1} \gets \text{FindAdjacentQueues}(L, Q)$
    \IF{$L \le Q_i.\text{max\_len} \times 1.10$}
        \STATE \textsc{AssignToQueue}(request, $Q_i$)
    \ELSIF{$L \ge Q_{i+1}.\text{min\_len} \times 0.90$}
        \STATE \textsc{AssignToQueue}(request, $Q_{i+1}$)
    \ELSE
        \STATE \COMMENT{True gap detected: Create Bubble Queue}
        \STATE $available \gets Q_{i+1}.\text{min\_len} - Q_i.\text{max\_len}$
        \STATE $range \gets \min(\text{config.default\_bubble\_width},\; available)$
        \STATE $new\_min \gets \max(L - range/2,\; Q_i.\text{max\_len})$
        \STATE $new\_max \gets \min(L + range/2,\; Q_{i+1}.\text{min\_len})$
        \STATE $Q_{\text{new}} \gets \text{CreateNewQueue}(new\_min,\; new\_max)$
        \STATE \textsc{AssignToQueue}(request, $Q_{\text{new}}$)
    \ENDIF
\end{algorithmic}
\end{algorithm}

\end{document}